\documentclass[letterpaper,11pt]{article}
\usepackage[utf8]{inputenc}
\usepackage[margin=1in]{geometry}

\usepackage{amsmath, amsthm, amssymb, thm-restate}
\usepackage{algorithmicx}
\usepackage[table,xcdraw]{xcolor}
\usepackage[ruled,vlined,linesnumbered]{algorithm2e}
\usepackage{setspace}
\usepackage{mathtools}

\usepackage{comment}
\usepackage[most]{tcolorbox} 
\usepackage{xfrac}
\usepackage{hyperref}
\usepackage{multirow}
\usepackage{caption}
\usepackage{bm}
\usepackage{newfloat}
\usepackage{enumitem}
\usepackage{dblfloatfix} 
\usepackage{wrapfig}

\usepackage{fancyhdr}

\usepackage{tikz}
\usetikzlibrary{shapes,fit,backgrounds,calc,decorations.pathmorphing}
\definecolor{myBlue}{HTML}{1f77b4}
\definecolor{myGreen}{HTML}{2ca02c}
\definecolor{myRed}{HTML}{d62728}
\definecolor{myBrown}{HTML}{8b4513}
\definecolor{SecondaryLightBlue}{HTML}{aec7e8}
\tikzset{
    bm/.style={
        decoration={snake, amplitude=0.4mm, segment length=4mm},
        decorate,
        line width=1pt,
        color=myGreen
    },
    bmtwo/.style={
        decoration={zigzag, amplitude=1mm, segment length=3mm},
        decorate,
        line width=1pt,
        color=myBrown
    }
}

\definecolor{mygreen}{RGB}{10,150,110}
\definecolor{myred}{RGB}{150,10,20}

\hypersetup{
     colorlinks=true,
     citecolor= mygreen,
     linkcolor= myred
}

\title{A $0.51$-Approximation of Maximum Matching \\ in Sublinear $n^{1.5}$ Time}
\author{Sepideh Mahabadi\thanks{Microsoft Research. E-mail: \email{smahabadi@microsoft.com}.} \and Mohammad Roghani\thanks{Stanford University. E-mail: \email{roghani@stanford.edu}. Work done while the author was an intern at Microsoft Research.} \and Jakub Tarnawski\thanks{Microsoft Research. E-mail: \email{jakub.tarnawski@microsoft.com}.}}
\date{}

\renewcommand{\epsilon}{\varepsilon}

\DeclareMathOperator{\E}{\ensuremath{\normalfont \textbf{E}}}

\newcommand{\hiddencomment}[1]{}

\newcommand{\GMM}[0]{\ensuremath{\textup{GMM}}}
\newcommand{\RGMM}[0]{\ensuremath{\textup{RGMM}}}
\newcommand{\wt}[1]{\ensuremath{\widetilde{#1}}}

\newcommand{\Touter}{\ensuremath{T_{\mathrm{outer}}}}
\newcommand{\Tinner}{\ensuremath{T_{\mathrm{inner}}}}
\newcommand{\costinner}{\ensuremath{\mathrm{cost}_{\mathrm{inner}}}}

\newcommand{\ceil}[1]{{\left\lceil{#1}\right\rceil}}

\DeclareMathOperator*{\Prob}{\ensuremath{\textnormal{Pr}}}
\renewcommand{\Pr}{\Prob}

\usepackage[noabbrev,nameinlink]{cleveref}
\crefname{lemma}{Lemma}{Lemmas}
\crefname{theorem}{Theorem}{Theorems}
\crefname{property}{Property}{Properties}
\crefname{claim}{Claim}{Claims}
\crefname{result}{Result}{Results}
\crefname{definition}{Definition}{Definitions}
\crefname{observation}{Observation}{Observations}
\crefname{proposition}{Proposition}{Propositions}
\crefname{assumption}{Assumption}{Assumptions}
\crefname{line}{Line}{Lines}
\crefname{figure}{Figure}{Figures}
\creflabelformat{property}{(#1)#2#3}
\crefname{equation}{}{}
\crefname{section}{Section}{Sections}
\crefname{appendix}{Appendix}{Appendices}
\crefname{algCounter}{Algorithm}{Algorithms}
\Crefname{algCounter}{Algorithm}{Algorithms}

\newtheorem{lemma}{Lemma}[section]
\newtheorem{theorem}[lemma]{Theorem}
\newtheorem{proposition}[lemma]{Proposition}

\newtheorem{claim}[lemma]{Claim}

\newtheorem{observation}[lemma]{Observation}

\newtheorem*{remark*}{Remark}

\definecolor{mylightgray}{RGB}{230,230,230}

\algnewcommand{\IIf}[2]{\textbf{if} #1 \textbf{then} #2}
\algnewcommand{\EndIIf}{\unskip\ \algorithmicend\ \algorithmicif}

\newenvironment{whitetbox}{
\par\addvspace{0.1cm}
\begin{tcolorbox}[width=\textwidth,
                  boxsep=5pt,
                  left=1pt,
                  right=1pt,
                  top=2pt,
                  bottom=2pt,
                  boxrule=1pt,
                  arc=0pt,
                  colframe=black,
                  colback=white
                  ]
}{
\end{tcolorbox}
}

\newcounter{algCounter}

\providecommand{\email}[1]{\href{mailto:#1}{\nolinkurl{#1}\xspace}}

\begin{document}

\maketitle

\begin{abstract}
We study the problem of estimating the size of a maximum matching in sublinear time. The problem has been studied extensively in the literature and various algorithms and lower bounds are known for it. Our result is a $0.5109$-approximation algorithm with a running time of $\tilde{O}(n\sqrt{n})$.
    
All previous algorithms either provide only a marginal improvement (e.g., $2^{-280}$) over the $0.5$-approximation that arises from estimating a \emph{maximal} matching, or have a running time that is nearly $n^2$. Our approach is also arguably much simpler than other algorithms beating $0.5$-approximation.
\end{abstract}

\section{Introduction}
Given a graph $G=(V,E)$, a matching is a set of edges with no common endpoints, and the maximum matching problem asks for finding a largest such subset. Matching is a fundamental combinatorial optimization problem,
and a benchmark for new algorithmic techniques in all major computational models.
It also has a wide range of applications such as ad allocation, social network recommendations, and information retrieval, among others. Given that many of these applications need to handle large volumes of data, the study of sublinear time algorithms for {\em estimating} the maximum matching size has received considerable attention over the past two decades.\footnote{It is impossible to \emph{find the edges} of any constant-approximate matching in time sublinear with respect to the size of the input \cite{ParnasRon07}.} A sublinear time algorithm is not allowed to read the entire graph, which would take $\Omega(n^2)$ time where $n=|V|$; instead it  is provided \emph{oracle access} to the input graph and must run in $o(n^2)$ time. There are two main oracles for graph problems considered in the literature, which we also consider in this work.

\begin{itemize}
\item \textbf{Adjacency list oracle.} Here, the algorithm can query $(v,i)$, where $v\in V$ and $i\leq n$, and the oracle reports the $i$-th neighbor of the vertex $v$ in its adjacency list, or NULL if $i$ is larger than the number of $v$'s neighbors.
\item \textbf{Adjacency matrix oracle.} Here, the algorithm can query $(u,v)$, where $u,v\in V$, and the oracle reports whether there exists an edge between $u$ and $v$.
\end{itemize}

Earlier results on estimating the size of the maximum matching in sublinear time mostly focused on graphs with bounded degree $\Delta$, starting with the pioneering work of Parnas and Ron \cite{ParnasRon07}, and later works of \cite{nguyen2008constant, YoshidaYI09, RubinfeldTVX11, AlonRVX12, LeviRY17, Ghaffari-FOCS22}. However, for $\Delta=\Omega(n)$ these do not lead to sublinear time algorithms.
Thus, later works focused on general graphs with arbitrary maximum degree and managed to obtain sublinear time algorithms for them \cite{chen2020sublinear,kapralov2020space,Behnezhad21,BehnezhadRRS23, bhattacharya2023sublinear,bhattacharya2023dynamic}. In particular, the state of the art results can be categorized into two regimes:
\begin{itemize}
\item Algorithms that run in slightly sublinear time, i.e., $n^{2-\Omega_{\epsilon}(1)}$. For example, the works of \cite{behnezhad2023sublinear, bhattacharya2023sublinear} gave a $(2/3-\epsilon)$-approximation algorithm that runs in time $n^{2-\Omega_{\epsilon}(1)}$. Later, \cite{bhattacharya2023dynamic} improved the approximation factor to $1-\epsilon$.
\item Algorithms whose approximation factor is $0.5+\epsilon$. The state of the art result in this category is the work of \cite{BehnezhadRRS23}, whose running time is $n^{1+\epsilon}$ for an approximation factor of $0.5+\Omega_{\epsilon}(1)$. However, the best approximation factor one can get using their trade-off is only $0.5+2^{-280}.$\footnote{More specifically, for $\epsilon \in (0,1/4)$, they get an algorithm with approximation factor of $0.5+2^{-70/\epsilon}$ that runs in time $O(n^{1+\epsilon})$.} Indeed, they mention: \textit{We do not expect our techniques to lead to a better than, say, $0.51$-approximation in $n^{2-\Omega(1)}$ time.}
\end{itemize}
Thus, all known algorithms either give only a marginal improvement over the $0.5$-approximation
arising from estimating the size of a \emph{maximal} matching
(which can be done even in $\wt{O}(n)$ time~\cite{Behnezhad21}),
or have a running time that is nearly $n^2$.
Making a significant improvement on both fronts simultaneously
has remained an elusive open question.

\paragraph{Our results.} In this work, we show how to beat the  factor $0.5$ in time that is strongly sublinear. Specifically, we present an algorithm that runs in time $O(n\sqrt{n})$ and achieves an approximation factor of $0.5109$ for estimating the size of the maximum matching.

It is worth noting that our algorithm is much simpler, both in terms of implementation and analysis,  compared to the prior works that beat $0.5$-approximation~\cite{BehnezhadRRS23}.

\begin{restatable}{theorem}{maintheoremAdjlist}\label{thm:adj-list-theorem}
There exists an algorithm that, given access to the adjacency list of a graph, estimates the size of the maximum matching with a multiplicative approximation factor of $0.5109$ and runs in $\widetilde{O}(n\sqrt{n})$ time with high probability.
\end{restatable}

\begin{restatable}{theorem}{maintheoremAdjmat}\label{thm:adj-mat-theorem}
There exists an algorithm that, given access to the adjacency matrix of a graph, estimates the size of the maximum matching with a multiplicative-additive approximation factor of $(0.5109, o(n))$ and runs in $\widetilde{O}(n\sqrt{n})$ time with high probability.
\end{restatable}

Moreover, our algorithm can be employed as a subroutine for Theorem 2 of \cite{Behnezhad23} to obtain an improved approximation ratio in the dynamic setting.
More precisely, that result requires a subroutine for estimating the maximum matching size in $\wt{O}(n\sqrt{n})$ time,
for which it uses the $0.5$-approximation of~\cite{Behnezhad21}.
Our algorithm can be used instead, resulting in a very slight improvement to the overall approximation guarantee for~\cite{Behnezhad23}.


We note that the framework of~\cite{bhattacharya2023dynamic} can also be used to obtain a similar result. Their algorithm, which performs a single iteration to find a constant fraction of augmenting paths of length three on top of a maximal matching, likewise yields a better-than-2 approximate matching with $n^{2 - \Omega(1)}$ running time. However, the trade-off in this approach is worse in terms of both the approximation ratio and the running time.

\paragraph{Related work.} On the lower bound front, Parnas and Ron \cite{ParnasRon07} demonstrated that any algorithm getting a constant approximation of the maximum matching size needs at least $\Omega(n)$ time. More recently, the work of \cite{behnezhad2023sublinear} established that any algorithm providing a $(2/3 + \Omega(1), \epsilon n)$-multiplicative-additive approximation requires at least $n^{6/5-o(1)}$ time. For sparse graphs, a lower bound of $\Delta^{\Omega(1/\epsilon)}$ was shown for any $\epsilon n$ additive approximation \cite{behnezhad2023local}. For dense graphs, \cite{behnezhad2024approximating} showed a lower bound of $n^{2-O_\epsilon(1)}$ for the runtime of algorithms achieving such additive approximations.

\paragraph{Paper organization.} In \Cref{sec:overview}, we provide an overview of the challenges encountered while designing our algorithm and the techniques used to address them. We first develop an algorithm for bipartite graphs with a multiplicative-additive error in \Cref{sec:bipartite}, avoiding additional challenges that arise from general graphs, trying to obtain multiplicative error (in the adjacency list model), or working with the adjacency matrix. In \Cref{sec:general}, we extend our algorithm to handle general graphs. In \Cref{sec:multiplicative}, we demonstrate how to achieve a multiplicative approximation guarantee. Finally, in \Cref{sec:matrix}, we present a simple reduction showing that our algorithm also works in the adjacency matrix model with a multiplicative-additive error.

\section{Technical Overview}\label{sec:overview}

In this section, we provide an overview of the techniques used in this paper to design our algorithm. We begin with the two-pass semi-streaming algorithm of Konrad and Naidu~\cite{KonradN21} for bipartite graphs. In the first pass, the algorithm constructs a maximal matching $M$. In the second pass, it constructs a maximal $b$-matching between vertices matched in $M$ and those unmatched in $M$. More specifically, each vertex in $V(M)$ has a capacity of $k$, while each vertex in $V \setminus V(M)$ has a capacity of $kb$, where $b = 1 + \sqrt{2}$ and $k$ is a large constant.
The idea is that if $M$ is far from maximum, the $b$-matching will contain many length-3 augmenting paths that can be used to augment $M$.
This algorithm obtains a $(2 - \sqrt{2}) \approx 0.585$-approximation.

Our goal is to develop a sublinear-time algorithm by translating this semi-streaming two-pass algorithm to the sublinear time model.
When trying to do so, several challenges arise.
In this section we describe them step by step, and show how to overcome them.

\paragraph{Challenge (1): constructing a maximal matching in sublinear time is not possible.} In fact, finding all edges of any constant-factor approximation of the maximum matching is impossible in sublinear time due to \cite{ParnasRon07}. Dynamic algorithms for maximum matching \cite{Behnezhad23, BhattacharyaKSW23} use the following approach: they maintain a maximal matching $M$ and then apply the sublinear-time \textbf{random greedy maximal matching (RGMM)} algorithm of Behnezhad~\cite{Behnezhad21} to estimate the size of the maximal $b$-matching. In our setting, we cannot afford to explicitly construct $M$. However, we can obtain oracle access to $M$ using the sublinear-time RGMM algorithm of \cite{Behnezhad21}. More specifically, we can query whether a vertex $v$ is matched in $M$ or not in $\wt{O}(n)$ time. Therefore, a possible solution to address the first challenge is to design two nested oracles: the outer oracle attempts to build a maximal $b$-matching, whereas the inner oracle checks the status of vertices (matched or not in $M$) to correctly filter edges and assign capacities to each vertex.

\paragraph{Challenge (2): two nested oracles require $\Omega(n^2)$ time.} The algorithm of \cite{Behnezhad21} runs in $\wt{O}(\bar{d}(G))$ time, where $\bar{d}(G)$ denotes the average degree of the graph $G$. Additionally, for the outer oracle, it requires $\wt{O}(\bar{d}(G[V(M), V \setminus V(M)]))$ time (i.e., queries to the inner oracle). Unfortunately, it is possible for both $\bar{d}(G)$ and $\bar{d}(G[V(M), V \setminus V(M)])$ to be as large as $\Omega(n)$.  For example, consider a graph with a vertex set $A \cup B$, where $|A| = |B| = n/2$. The edges within $A$ form a complete bipartite graph, while there is an $\epsilon n/2$-regular graph between $A$ and $B$. After running the RGMM algorithm, most edges in the maximal matching belong to $G[A]$, and most vertices in $B$ are unmatched. Consequently, we have $\bar{d}(G[V(M), V \setminus V(M)]) = \Omega(n)$. 

To address this issue, we sparsify the original graph while manually constructing a matching~$M$. In a preprocessing step, starting from an empty $M$, for each unmatched vertex in the graph, we sample $\wt{\Theta}(\sqrt{n})$ neighbors uniformly at random. If an unmatched neighbor exists, we match the two vertices and add this edge to $M$.  Using this preprocessing step, we show that after spending $\wt{O}(n\sqrt{n})$ time, the induced subgraph of vertices that remain unmatched in $M$ has a maximum degree of $\sqrt{n}$ with high probability. Moreover,
since we explicitly materialize $M$,
we are able to check if any vertex is matched in $M$ in $O(1)$ time, eliminating the need for costly oracle calls.
Note that $M$ need not be maximal in $G$; therefore we next extend it to a maximal matching.

Let $M'$ be a maximal matching in $G[V \setminus V(M)]$ obtained by running the sublinear time RGMM algorithm of \cite{Behnezhad21}. Now $M \cup M'$ is a maximal matching for $G$. Inspired by the two-pass semi-streaming algorithm of \cite{KonradN21}, we attempt to augment the maximal matching $M \cup M'$ in two possible ways (see also \cref{fig1}):

\begin{enumerate}
    \item Augment $M'$ using a $b$-matching between $V(M')$ and $V \setminus V(M) \setminus V(M')$. The algorithm then outputs the size of the augmented matching, plus the size of the previously constructed matching $M$.  
    \item Augment $M$ using a $b$-matching between $V(M)$ and $V \setminus V(M)$. The algorithm then outputs the size of the augmented matching.  
\end{enumerate}

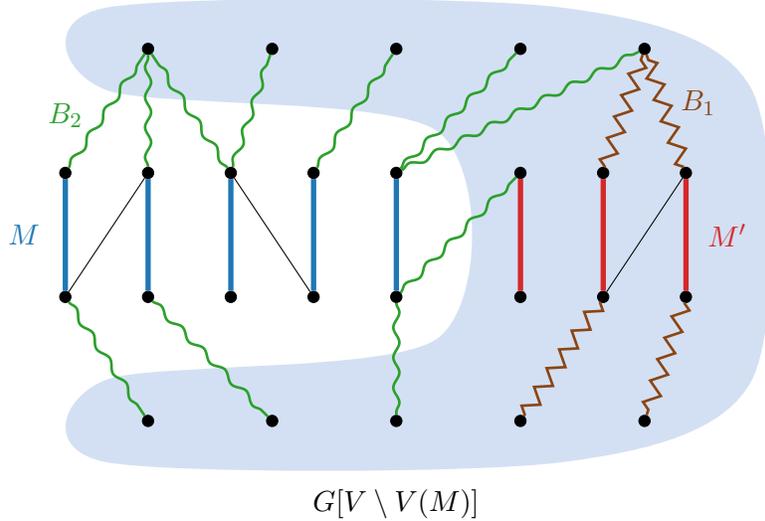
\begin{figure}[t]
    \centering
    \begin{tikzpicture}[scale=1.1, every node/.style={circle, draw, fill=black, inner sep=1.5pt}]

    \node (b1) at (1,0) {};
    \node (c1) at (2,0) {};
    \node (d1) at (3,0) {};
    \node (e1) at (4,0) {};
    \node (f1) at (5,0) {};
    \node (b2) at (1,1.5) {};
    \node (c2) at (2,1.5) {};
    \node (d2) at (3,1.5) {};
    \node (e2) at (4,1.5) {};
    \node (f2) at (5,1.5) {};
    \draw[line width=2pt, myBlue] (b1) -- (b2);
    \draw[line width=2pt, myBlue] (c1) -- (c2);
    \draw[line width=2pt, myBlue] (d1) -- (d2);
    \draw[line width=2pt, myBlue] (e1) -- (e2);
    \draw[line width=2pt, myBlue] (f1) -- (f2);

\node[draw=none, fill=none, color=myBlue] at (0.5, 0.75) {$M$};
\node[draw=none, fill=none, color=myRed] at (9, 0.75) {$M'$};
\node[draw=none, fill=none, color=myBrown] at (8.65, 2.35) {$B_1$};
\node[draw=none, fill=none, color=myGreen] at (1, 2.2) {$B_2$};
\node[draw=none, fill=none] at (5, -2.5) {$G[V \setminus V(M)]$};
    
    \draw (b1) -- (c2);
    \draw (e1) -- (d2);

    \node (up1) at (2,3) {};
    \node (up2) at (3.5,3) {};
    \node (up3) at (5,3) {};
    \node (up4) at (6.5,3) {};
    \node (up5) at (8,3) {};
    \node (do1) at (2,-1.5) {};
    \node (do2) at (3.5,-1.5) {};
    \node (do3) at (5,-1.5) {};
    \node (do4) at (6.5,-1.5) {};
    \node (do5) at (8,-1.5) {};

    \draw[bm] (b2) -- (up1);
    \draw[bm] (c2) -- (up1);
    \draw[bm] (d2) -- (up1);
    \draw[bm] (d2) -- (up2);
    \draw[bm] (e2) -- (up3);
    \draw[bm] (f2) -- (up4);
    \draw[bm] (f2) -- (up5);

    \draw[bm] (b1) -- (do1);
    \draw[bm] (c1) -- (do2);
    \draw[bm] (f1) -- (do3);

    \node (g1) at (6.5,0) {};
    \node (h1) at (7.5,0) {};
    \node (i1) at (8.5,0) {};
    \node (g2) at (6.5,1.5) {};
    \node (h2) at (7.5,1.5) {};
    \node (i2) at (8.5,1.5) {};
    \draw[line width=2pt, myRed] (g1) -- (g2);
    \draw[line width=2pt, myRed] (h1) -- (h2);
    \draw[line width=2pt, myRed] (i1) -- (i2);

    \draw[bm] (f1) -- (g2);

    \draw (h1) -- (i2);

    \draw[bmtwo] (i1) -- (do5);
    \draw[bmtwo] (h1) -- (do4);
    \draw[bmtwo] (h2) -- (up5);
    \draw[bmtwo] (i2) -- (up5);

    \begin{scope}[on background layer]
    \fill[SecondaryLightBlue!55] plot [smooth cycle] coordinates {(9.3,2.5) (7,3.5) (1.5,3.5) (1.5,2.5) (5.5,2) (5.5,-0.5) (1.5,-1) (1.5,-2) (7,-2) (9.3,-1)};
    \end{scope}
\end{tikzpicture}
\vspace{-2em}
\caption{Our algorithm explicitly constructs a matching $M$ (blue), which need not be maximal in~$G$.
We extend it with another matching $M'$ (red), such that $M \cup M'$ is maximal.
The highlighted (light blue) subgraph $G[V \setminus V(M)]$ has degree at most $\sqrt{n}$ with high probability.
In case 1, our algorithm augments $M'$ using a $b$-matching $B_1$ (zigzag edges, brown).
In case 2, our algorithm augments $M$ using a $b$-matching $B_2$ (swirly edges, green).%
}
    \label{fig1}
\end{figure}

The key intuition here (think of the case when $M \cup M'$ yields only a $0.5$-approximation) is that either $M'$ is sufficiently large, making $|M| + (2-\sqrt{2})\cdot 2|M'|$ larger than the approximation guarantee, or $M$ itself is large enough so that $(2-\sqrt{2})\cdot 2|M|$ meets our approximation guarantee (note that $(2-\sqrt{2})$ is the approximation guarantee of \cite{KonradN21}). Augmenting $M$ using a $b$-matching is easier since we have explicit access to $M$ and only need to run a single RGMM oracle to estimate the size of the $b$-matching. Our first estimate, which requires finding a $b$-matching between $V(M')$ and $V \setminus V(M) \setminus V(M')$, is more challenging since we do not have explicit access to $M'$.
To avoid the $\Omega(n^2)$ running time of the two nested oracles, we make crucial use of the aforementioned property that the subgraph of vertices unmatched in $M$ has low induced degree (at most $\sqrt{n}$);
this is the reason why we only try to augment $M'$ rather than $M \cup M'$.
We will discuss this in the next  paragraphs.

\paragraph{Challenge (3): the algorithm does not have access to the adjacency list of $G[V \setminus V(M)]$.} After the sparsification step, the average degree $d$ of $G[V \setminus V(M)]$ is at most $\sqrt{n}$. Hence, if the algorithm had access to the adjacency list of $G[V \setminus V(M)]$, it could run the nested oracles in $\wt{O}(d^2) = \wt{O}(n)$ time by executing two RGMM algorithms: inner oracle for computing $M'$ and outer oracle for the $b$-matching to augment $M'$. But, since the nested oracles may visit up to $n$ vertices, and retrieving the full adjacency list of a vertex in $G[V \setminus V(M)]$ requires $\Omega(n)$ time, it seems that the overall running time of the algorithm could still be as high as $\Omega(n^2)$.

Here, we leverage two key properties of the RGMM algorithm to refine the runtime analysis. The first property is that at each step, the algorithm requires only a random neighbor of the currently visited vertex. Intuitively, if a vertex has degree $\Theta(n)$ in $G$, in expectation it takes $O(n/d)$ samples from the adjacency list of the original graph to encounter a vertex from $G[V \setminus V(M)]$. Thus, if all vertices in $G[V \setminus V(M)]$ had degree $d$, one could easily argue that the running time of the algorithm is $\wt{O}(d^2 \cdot n/d) = \wt{O}(n\sqrt{n})$. However, vertex degrees can vary, and for a vertex with a constant degree, we would need $\Omega(n)$ samples from the adjacency list of $G$ to find a single neighbor in $G[V \setminus V(M)]$. To address this challenge, we utilize another property of the RGMM algorithm, recently proven by \cite{steiner-tree-itcs}. Informally, this result shows that
during oracle calls for RGMM,
a vertex is visited proportionally to its degree, implying that low-degree vertices are visited only a small number of times.

\paragraph{Challenge (4): outer oracle creates biased inner oracle queries.} The final main challenge we discuss here is that the simple $\wt{O}(n\sqrt{n})$ bound, which we informally proved in the previous paragraph, relies on the tacit assumption that the inner oracle queries generated by the outer oracle correspond to $\wt{O}(\sqrt{n})$ uniformly random calls to the inner oracle. Indeed, the running time of the algorithm of \cite{Behnezhad21} is analyzed for a uniformly random query vertex; however, there may exist a vertex $v$ in the graph for which calling the inner oracle takes significantly more than $\wt{O}(d)$ time. Consequently, if all outer oracle calls end up querying $v$, the running time could be significantly worse than $\wt{O}(n\sqrt{n})$.  To overcome this issue, we use the result of \cite{steiner-tree-itcs} along with the fact that the maximum degree of $G[V \setminus V(M)]$ is $\wt{O}(\sqrt{n})$. We show that for any vertex $v$, the outer oracle queries the inner oracle for $v$ at most $\wt{O}(\deg_{G[V \setminus V(M)]}(v) / \sqrt{n})$ times in expectation. This enables us to formally prove that the total running time of the algorithm remains at most $\wt{O}(n\sqrt{n})$.

\paragraph{General graphs and the adjacency matrix model.} There are additional challenges when dealing with general graphs as opposed to bipartite graphs, such as the fact that the sizes of the maximal matching and the $b$-matching alone are insufficient to achieve a good approximation ratio. For general graphs, our algorithm estimates the maximum matching in the union of the maximal matching and the $b$-matching, which requires using the $(1-\epsilon)$-approximate local computation algorithm (LCA) by \cite{LeviRY17} on the subgraph formed by this union, to which we only have oracle access. We encourage readers to refer to \Cref{sec:general} for more details about the techniques used there.

Additionally, for more information on the extension of the algorithm that operates in the adjacency matrix model, we recommend readers to check \Cref{sec:matrix}.

\section{Preliminaries}\label{sec:preliminaries}
Given a graph $G$, we use $V(G)$ to denote its set of vertices and use $E(G)$ to refer to its set of edges.
For a vertex $v\in V(G)$, we use $\deg_G(v)$ to denote the degree of the vertex, i.e., the number of edges with one endpoint equal to $v$. We use $\Delta(G)$ to denote the maximum degree over all vertices in the graph, and $\bar{d}(G)$ to denote the average degree of the graph. Further, we use $\mu(G)$ to denote the size of the maximum matching in $G$.

Given a graph $G=(V,E)$, and a subset of vertices $A\subseteq V$, $G[A]$ is defined to be the induced subgraph consisting of all edges with both endpoints in $A$.
Further, given disjoint subsets $A,B\subset V$ of vertices, $G[A,B]$ is defined to be the bipartite subgraph of $G$ consisting of all edges between $A$ and $B$.

Given a matching $M$ in $G$, an {\em augmenting path} is a simple path starting and ending at different vertices, such that the first and the last vertices are unmatched in $M$, and the edges of the path alternate between not belonging to $M$ and belonging to $M$.

Given a vector $b$ of integer capacities of dimension $|V(G)|$, a {\em $b$-matching} in $G$ is a {\em multi-set} $F$ of edges in $G$ such that each vertex $v\in V$ appears no more than $b_v$ times as an endpoint of an edge in $F$. 

Given a graph $G$ and a permutation $\pi$ over its edges, $\GMM(G, \pi)$ is used to refer to the unique matching $M$ obtained by the following process. We let $M$ be initialized as empty, and consider the edges of $G$ one by one according to the permutation $\pi$. We add an edge to the matching $M$ if none of its endpoints are already matched in $M$. A random greedy maximal matching, i.e., $\RGMM(G)$ is the 
matching obtained by picking a permutation $\pi$ uniformly at random and outputting $\GMM(G,\pi)$.

\begin{proposition}[\cite{Behnezhad21}]\label{prop:rgmm}
    There exists an algorithm that,
    given adjacency list access
    to a graph~$G$ of average degree $d$,
    for a random vertex $v$ and a random permutation $\pi$,
    determines if $v$ is matched in $\GMM(G, \pi)$ in $\widetilde{O}(d)$ expected time. 
\end{proposition}

Given the problem of maximizing a function $f:D\rightarrow \mathbb{R}$ defined over a domain $D$, with optimal value $f^*$, an $(\alpha,\beta)$-multiplicative-additive approximation of $f^*$ is a solution $s\in D$ such that $(f^*/\alpha) -\beta \leq f(s)\leq f^*$.

\section{Algorithm for Bipartite Graphs}\label{sec:bipartite}

We begin by describing our algorithm for bipartite graphs. We focus on implementing an algorithm with a multiplicative-additive approximation guarantee. Also, we assume that we have access to the adjacency list of the graph. These assumptions will help us avoid certain complications and challenges that arise when working with general graphs, the adjacency matrix model, or when trying to obtain a multiplicative approximation guarantee. To lift these assumptions, we can leverage strong tools and methods from the literature, which, with slight modifications, can be applied here. This section contains the main novelties of our approach and proofs. Our algorithm for bipartite graphs can be seen as a translation and implementation of a two-pass streaming algorithm, which we discuss in the next subsection.

\subsection{Two-Pass Streaming Algorithm for Bipartite Graphs}

Our starting point is the two-pass streaming algorithm which is described in \Cref{alg:two-pass}. This algorithm, or its variations, has appeared in previous works on designing streaming or dynamic algorithms for maximum matching \cite{KonradN21, BhattacharyaKSW23, Behnezhad23}. In words, the first pass of the algorithm only finds a maximal matching $M$. In the second pass, the algorithm finds a maximal $b$-matching $B$ in $G[V(M), V \setminus V(M)]$, where $V(M)$ is the set of vertices matched  by $M$. The capacities of vertices in $V(M)$ and in $V \setminus V(M)$ for the $b$-matching are $k$ and $kb$, respectively. Moreover, in the second pass of the algorithm, when an edge $(u,v)$ arrives in the stream, we add multiple copies of the edge to the subgraph $B$, as long as doing so does not violate the capacity constraints.

\begin{algorithm}
\caption{Two-pass Streaming Algorithm for Bipartite Graphs}
\label{alg:two-pass}

\textbf{Parameter:} let $b = 1 + \sqrt{2}$ and $k$ be an integer larger than $\frac{1}{b\epsilon^3}$.

\textbf{First Pass:} $M \gets$ maximal matching of $G$.  \algorithmiccomment{Finding maximal matching}

\textbf{Second Pass:} \algorithmiccomment{Finding $b$-matching}

Let $B = \emptyset$.
    
\For{$(u, v) \in G[V(M), \overline{V(M)}]$ where $u \in V(M)$}{
    \While{$\deg_B(u) < k$ and $\deg_B(v) < \ceil{kb}$}{
        $B \gets B \cup {(u, v)}$.  \algorithmiccomment{We allow multi edges}
    }

}

\Return $(1-1/b) \cdot |M| + 1/(kb) \cdot |B|$.

\end{algorithm}

Intuitively, the algorithm tries to find length-3 augmenting paths using the $b$-matching that it finds in the second pass. The following lemma shows the approximation guarantee of \Cref{alg:two-pass}.

\begin{lemma}[Lemma 3.3  in~\cite{BhattacharyaKSW23}] \label{lem:585bha}
    For any $\epsilon \in (0, 1)$, the output of \Cref{alg:two-pass} is a $(2 -\sqrt{2} - \epsilon)$-approximation for maximum matching of $G$.
\end{lemma}

\subsection{Sublinear Time Implementation of \Cref{alg:two-pass}}

In this section, we demonstrate how to implement a modification of \Cref{alg:two-pass} in the sublinear time model,  as outlined in \Cref{sec:overview}.

\paragraph{Sparsification:} In order to be able to use two levels of recursive oracle calls, we need to sparsify the graph. We first sample $\widetilde{O}(n \sqrt{n})$ edges and construct a maximal matching on the sampled edges to sparsify the induced subgraph of unmatched vertices. This sparsification step is formalized in \Cref{alg:sparsification}. In \Cref{lem:sparsification-maximal}, we show that if we sample enough edges, then vertices that remain unmatched after this phase have an induced degree of $\sqrt{n}$. This step is very similar to the algorithm of~\cite[Appendix A]{chen2020sublinear} for approximating a maximal matching.

\begin{algorithm}
\caption{Sparsification of the Induced Subgraph of Unmatched Vertices}
\label{alg:sparsification}
\textbf{Parameter:} let $c = 2 \sqrt{n} \cdot \log n$ be the sparsification parameter that controls the number of edges that the algorithm samples.

$M \gets \emptyset$

\For{$v \in V$}{
    \If{$v \notin V(M)$}{
        Sample $c$ vertices $u_1, \ldots, u_c$ from $N(v)$.
        
        \For{$i \gets 1 \ldots c$}{
            \If{$u_i \notin V(M)$}{
                $M \gets M \cup \{(v, u_i)\}$

                \textbf{break;} 
                
            }
        }
    }

}

\Return $M$

\end{algorithm}

\begin{claim}\label{clm:sparsification-time}
    \Cref{alg:sparsification} runs in $\widetilde{O}(n \sqrt{n})$ time.
\end{claim}
\begin{proof}
    For each vertex $v$ in the graph, the algorithm samples $\widetilde{O}(\sqrt{n})$ vertices from the adjacency list of the vertex $v$ if it is not matched by the time the algorithm processes the vertex in Line 3. Thus, the total running time is at most $\widetilde{O}(n\sqrt{n})$.
\end{proof}

\begin{lemma}\label{lem:sparsification-maximal}
    With high probability, we have $\Delta(G[V \setminus V(M)]) \leq \sqrt{n}$.
\end{lemma}
\begin{proof}
    We will show that for every $v \in V$, the probability that $v \in V \setminus V(M)$ and the degree of $v$ in $G[V \setminus V(M)]$ is larger than $\sqrt{n}$
    is at most $1/n^2$. The lemma then follows by a union bound over all $v \in V$.

    Consider the moment before $v$ is processed. Assume that at this time, still $v \in V \setminus V(M)$ and the degree of $v$ in $G[V \setminus V(M)]$ is larger than $\sqrt{n}$. Then, each of the $c$ samples has a probability of at least $\sqrt{n}/n$ to be one of the unmatched neighbors, in which case $v$ will become matched. Thus the probability that $v$ remains unmatched after it is processed is at most
    \begin{align*}
        \left(1 - \frac{1}{\sqrt{n}}\right)^c = \left(1 - \frac{1}{\sqrt{n}}\right)^{2\sqrt{n} \log n} \leq \left( \frac{1}{e}\right)^{2\log n} = \frac{1}{n^2}.
    \end{align*}
\end{proof}

\paragraph{Augmenting $M$ using nested oracles:} Now we are ready to present our sublinear algorithm. After sparsifying the graph by finding a partially maximal matching $M$, we try to augment $M$ in two different ways that
we have outlined in \Cref{sec:overview}
and which
are formalized in \Cref{alg:sublinear-first}.
See also \cref{fig1}.

For simplicity, we pretend that $kb \in \mathbb{Z}$ throughout the paper. Since $k$ is an arbitrarily large constant, using $kb$ instead of $\lceil kb \rceil$ leads to an arbitrarily small error in the calculations.
We also note that a maximal $b$-matching can be viewed as maximal matching if we duplicate vertices multiple times.

First, we try to augment the matching by designing a maximal matching oracle for $G[V \setminus V(M)]$ vertices and then another oracle for finding a $b$-matching between the vertices newly matched using the new oracle and unmatched vertices. Let $M'$ be the maximal matching of $G[V \setminus V(M)]$ that can be obtained by the oracle. We try to augment it with a $b$-matching $B_1$.

However, it is also possible that $|M'|$ is small compared to $|M|$, which implies that in the previous case, the $b$-matching does not help to find many augmenting paths, as the size of the maximal matching that we try to augment is too small. To account for this case, the algorithm also finds a $b$-matching $B_2$ between $V(M)$ and $V \setminus V(M)$. 

Note that because the algorithm finds the initial matching $M$ explicitly, checking whether a vertex belongs to $V(M)$ or not can be done in $O(1)$ time.

\begin{algorithm}[H]
\caption{Sublinear Time Algorithm for Bipartite Graphs with Access to the Adjacency List (see \cref{fig1})}
\label{alg:sublinear-first}

Run \Cref{alg:sparsification} with $c = 2\sqrt{n} \log n$ and let $M$ be its output.

Let $\mu_{M'}$ and $\mu_{B_1}$ be the estimate of the size of a random greedy maximal matching $M'$ in $G[V \setminus V(M)]$ and the estimate of the size of a random greedy maximal $b$-matching $B_1$ in $G[V(M'), V \setminus V(M) \setminus V(M')]$ by running \Cref{alg:first-case}. \Comment{Case 1}

Let $\mu_1 := |M| + (1 - \frac{1}{b}) \mu_{M'} + \frac{1}{kb} \mu_{B_1}$. \Comment{Case 1}

Let $\mu_{B_2}$ be the estimate of the size of a random greedy maximal $b$-matching $B_2$ in $G[V(M), V \setminus V(M)]$ by running \Cref{alg:second-case}. \Comment{Case 2}

Let $\mu_2 := (1 - \frac{1}{b}) |M| + \frac{1}{kb} \mu_{B_2}$. \Comment{Case 2}

\Return $\max (\mu_1, \mu_2)$.

\end{algorithm}

\begin{algorithm}[H]
\caption{Algorithm for the First Case}
\label{alg:first-case}
Let $b = 1 + \sqrt{2}$ and $k$ be an integer larger than $\frac{1}{b\epsilon^3}$.

Let $\pi$ be a random permutation over edges of $G[V \setminus V(M)]$ and let $M'$ be its corresponding random greedy maximal matching.

Let $G_1 := G[A,B]$ where $A = V(M')$ and $B = V \setminus V(M) \setminus V(M')$. 

Let $G'_1$ be a bipartite graph obtained from $G_1$ by adding $k$ copies of vertices in $A$ and $kb$ copies of vertices in $B$. Further, if there exists an edge between $u \in A$ and $v \in B$ in $G_1$, we add edges between all copies of $u$ and $v$ in $G'_1$.

$r \gets 6 \log^3 n$.

Run the algorithm of \Cref{prop:rgmm} for $r$ random vertices and fixed permutation $\pi$ in $G[V \setminus V(M)]$ and let $X_i$ be the indicator if the $i$-th vertex is matched.

Let $X \gets \sum_{i=1}^r X_i$ and $\mu_{M'} \gets \frac{nX}{2r} - \frac{n}{2\log n}$.

Run nested oracles of \Cref{prop:rgmm} for $r$ random vertices and fixed permutation $\pi$ in $G'_1$ and let $Y_i$ be the indicator if the $i$-th vertex is matched.

Let $Y \gets \sum_{i=1}^r Y_i$ and $\mu_{B_1} \gets \frac{nY}{2r} - \frac{n}{2\log n}$.

\Return $\mu_{M'}$ and $\mu_{B_1}$.

\end{algorithm}

\begin{algorithm}[H]
\caption{Algorithm for the Second Case}
\label{alg:second-case}
Let $b = 1 + \sqrt{2}$ and $k$ be an integer larger than $\frac{1}{b\epsilon^3}$.

Let $G_2 := G[A,B]$ where $A = V(M)$ and $B = V \setminus V(M)$.

Let $G'_2$ be a bipartite graph obtained from $G_2$ by adding $k$ copies of vertices in $A$ and $kb$ copies of vertices in $B$. Further, if there exists an edge between $u \in A$ and $v \in B$ in $G_2$, we add edges between all copies of $u$ and $v$ in $G'_2$.

$r \gets 6 \log^3 n$.

Run the algorithm of \Cref{prop:rgmm} for $r$ random vertices and permutations in $G'_2$ and let $Z_i$ be the indicator that shows if the $i$-th vertex is matched.

Let $Z \gets \sum_{i=1}^r Z_i$ and $\mu_{B_2} \gets \frac{nZ}{2r} - \frac{n}{2\log n}$.

\Return $\mu_{B_2}$.
    
\end{algorithm}

\paragraph{Implementation details of the algorithm:} There are some technical details in the implementation of the algorithm that are not included in the pseudocode:

\begin{itemize}
    \item \textbf{Access to the adjacency list of an induced subgraph:} Both in \Cref{alg:first-case} and \Cref{alg:second-case}, we run the algorithm of \Cref{prop:rgmm} for some induced subgraph of $G$ (for example, line 5 of \Cref{alg:second-case}). However, \Cref{prop:rgmm} works with access to the adjacency list of the input graph. To address this issue, we leverage an important property of the algorithm in \Cref{prop:rgmm}, namely that it only needs to find a random neighbor of a given vertex at each step of its execution. Now, whenever the algorithm requires a random neighbor of vertex $v$ in a subgraph $H$, it queries random neighbors in the original graph $G$ until it finds one that belongs to $H$. This increases the running time of the algorithm, as it may take $\omega(1)$ time to locate a valid neighbor in $H$, which we will formally bound in our runtime analysis.

    \item \textbf{Nested oracles in line 8 of \Cref{alg:first-case}:} Unlike $M$, we do not explicitly construct the maximal matching $M'$ in \Cref{alg:first-case}. Moreover, the edges of the subgraph $G_1'$ connect vertices matched by $M'$ with those that remain unmatched in either $M$ or $M'$. Hence, to verify whether an edge belongs to $G_1'$, we need to determine whether its endpoints are matched or unmatched in $M'$ by accessing the algorithm of \Cref{prop:rgmm}. This again increases the algorithm's runtime, which we will also formally bound in our runtime analysis.
\end{itemize}

\subsection{Analysis of the Approximation Ratio}

The following lemma,
an analogue of Observation~3.1 in~\cite{BhattacharyaKSW23},
substantiates the soundness of the estimates $\mu_1$ and $\mu_2$ produced in \Cref{alg:sublinear-first}.

\begin{lemma} \label{lem:fractional_matching}
    Let $M$, $M'$, $B_1$ and $B_2$ be as in the description of \Cref{alg:sublinear-first}. Then
    \begin{itemize}
        \item $\mu(G) \ge \mu(M \cup M' \cup B_1) \ge |M| + (1 - \frac{1}{b}) |M'| + \frac{1}{kb} |B_1|$,
        \item $\mu(G) \ge \mu(M \cup B_2) \ge (1 - \frac{1}{b}) |M| + \frac{1}{kb} |B_2|$.
    \end{itemize}
\end{lemma}
\begin{proof}
    Since $G$ is bipartite,
    by integrality of the bipartite matching polytope,
    it is enough to exhibit a fractional matching $x$ of the appropriate value $\sum_e x_e$. For case 1, we set
    \[
    x_e = \begin{cases}
        1 & e \in M, \\
        1 - \frac{1}{b} & e \in M', \\
        \frac{1}{kb} & e \in B_1.
    \end{cases}
    \]
    Note that $M$, $M'$ and $B_1$ are pairwise disjoint, and $B_1$ has no edge to $V(M)$. Therefore it is easy to verify that the degree constraints for $x$ are satisfied.
    For case 2, we similarly set
    \[
    x_e = \begin{cases}
        1 - \frac{1}{b} & e \in M, \\
        \frac{1}{kb} & e \in B_2.
    \end{cases}%
    \]%
\end{proof}

\newcommand{\mo}{\ensuremath{|M_1^*|}}
\newcommand{\mop}{\ensuremath{|M_1^{'*}|}}
\newcommand{\mt}{\ensuremath{|M_2^{*}|}}
\newcommand{\mtp}{\ensuremath{|M_2^{'*}|}}
\newcommand{\mtpp}{\ensuremath{|M_2^{''*}|}}

The following lemma states the $(2-\sqrt{2})$-approximation guarantee
of the "maximal matching plus $b$-matching" approach
obtained in prior work,
for both bipartite and general graphs.
We will invoke it for appropriate subgraphs of $G$
to obtain our guarantee.

\begin{lemma}
    \label{lem:585}
    Let $G'$ be a graph,
    $M'$ be any maximal matching in $G'$,
    and $B$ be a maximal $b$-matching in $G'[V(M'), V(G') \setminus V(M')]$
    for vertex capacities $k$ for vertices in $V(M')$
    and $kb$ for vertices in $V(G') \setminus V(M')$,
    where $k > \frac{1}{b \epsilon^3}$
    and $b = 1 + \sqrt{2}$.
    Then:
    \begin{itemize}
        \item for bipartite $G'$, we have $\mu(M' \cup B) \ge (1-\frac{1}{b})|M'|+\frac{1}{kb}|B| \ge (2-\sqrt{2}-\epsilon)\mu(G')$,
        \item 
        for general $G'$, if $B$ is a \emph{random greedy} maximal $b$-matching,
        we still have $\E[\mu(M' \cup B)] \ge (2-\sqrt{2}-\epsilon)\mu(G')$.
    \end{itemize}
\end{lemma}
\begin{proof}
    The first statement is the same as \cref{lem:585bha}, and shown as Lemma 3.3 in~\cite{BhattacharyaKSW23}.
    The second statement is shown as Claim 5.5 in \cite{AzarmehrBR24}.
\end{proof}

The following lemma is the crux of our approximation ratio analysis.

\begin{lemma} \label{lem:apx051bipartite}
    In a bipartite graph $G$,
    let $M$, $M'$, $B_1$ and $B_2$ be as in the description of \Cref{alg:sublinear-first}. Then
    \[\max \left[|M| + (1 - \frac{1}{b}) |M'| + \frac{1}{kb} |B_1|, (1 - \frac{1}{b}) |M| + \frac{1}{kb} |B_2| \right] \geq 0.5109  \cdot\mu(G) .\]
\end{lemma}
\begin{proof}
    Fix a maximum matching $M^*$ in $G$,
    and partition its edges as follows (see \cref{fig2}):
    \begin{itemize}
        \item $M_2^*$ are those with both endpoints in $V(M)$,
        \item $M_2^{'*}$ are those with one endpoint in $V(M)$ and the other in $V(M')$,
        \item $M_2^{''*}$ are those with both endpoints in $V(M')$,
        \item $M_1^*$ are those with one endpoint in $V(M)$ and the other in $V \setminus V(M) \setminus V(M')$,
        \item $M_1^{'*}$ are those with one endpoint in $V(M')$ and the other in $V \setminus V(M) \setminus V(M')$.
    \end{itemize}

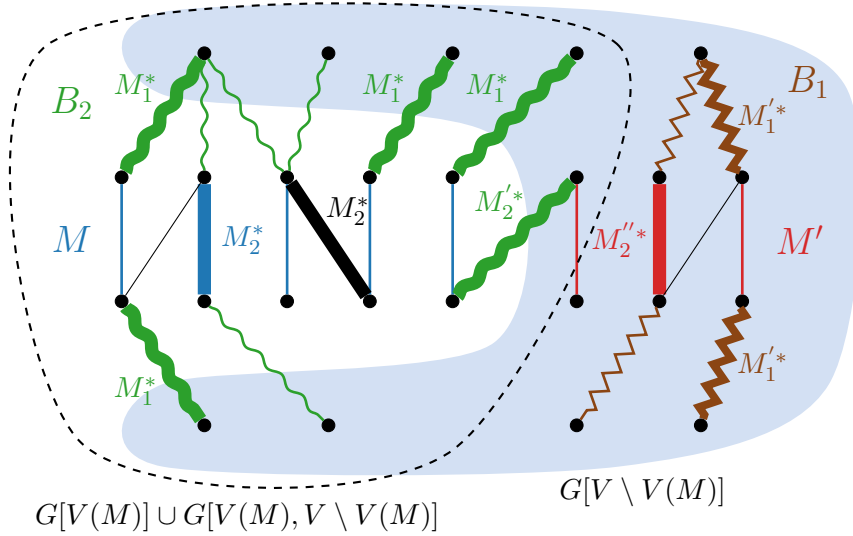
\begin{figure}[t]
    \centering
    \begin{tikzpicture}[scale=1.1, every node/.style={circle, draw, fill=black, inner sep=1.7pt}]

    \node (b1) at (1,0) {};
    \node (c1) at (2,0) {};
    \node (d1) at (3,0) {};
    \node (e1) at (4,0) {};
    \node (f1) at (5,0) {};
    \node (b2) at (1,1.5) {};
    \node (c2) at (2,1.5) {};
    \node (d2) at (3,1.5) {};
    \node (e2) at (4,1.5) {};
    \node (f2) at (5,1.5) {};
    \draw[line width=1pt, myBlue] (b1) -- (b2);
    \draw[line width=5pt, myBlue] (c1) -- (c2) node[draw=none, fill=none, midway, right] {$M_2^*$};
    \draw[line width=1pt, myBlue] (d1) -- (d2);
    \draw[line width=1pt, myBlue] (e1) -- (e2);
    \draw[line width=1pt, myBlue] (f1) -- (f2);

\node[draw=none, fill=none, color=myBlue, font=\Large] at (0.4, 0.75) {$M$};
\node[draw=none, fill=none, color=myRed, font=\Large] at (9.2, 0.75) {$M'$};
\node[draw=none, fill=none, color=myBrown, font=\Large] at (9.3, 2.65) {$B_1$};
\node[draw=none, fill=none, color=myGreen, font=\Large] at (0.4, 2.4) {$B_2$};
\node[draw=none, fill=none] at (7.3, -2.3) {$G[V \setminus V(M)]$};
\node[draw=none, fill=none] at (2.4, -2.6) {$G[V(M)] \cup G[V(M), V \setminus V(M)]$};
    
    \draw (b1) -- (c2);
    \draw[line width=5pt] (e1) -- (d2) node[draw=none, fill=none, pos=0.55, inner sep=0pt, above right] {$M_2^*$};

    \node (up1) at (2,3) {};
    \node (up2) at (3.5,3) {};
    \node (up3) at (5,3) {};
    \node (up4) at (6.5,3) {};
    \node (up5) at (8,3) {};
    \node (do1) at (2,-1.5) {};
    \node (do2) at (3.5,-1.5) {};
    \node (do4) at (6.5,-1.5) {};
    \node (do5) at (8,-1.5) {};

    \draw[bm,line width=5pt] (b2) -- (up1) node[draw=none, fill=none, pos=0.5, above left] {$M_1^*$};
    \draw[bm] (c2) -- (up1);
    \draw[bm] (d2) -- (up1);
    \draw[bm] (d2) -- (up2);
    \draw[bm,line width=5pt] (e2) -- (up3)  node[draw=none, fill=none, pos=0.5, above left] {$M_1^*$};
    \draw[bm,line width=5pt] (f2) -- (up4)  node[draw=none, fill=none, pos=0.5, above left] {$M_1^*$};

    \draw[bm,line width=5pt] (b1) -- (do1) node[draw=none, fill=none, pos=0.5, below left] {$M_1^*$};
    \draw[bm] (c1) -- (do2);


    \node (g1) at (6.5,0) {};
    \node (h1) at (7.5,0) {};
    \node (i1) at (8.5,0) {};
    \node (g2) at (6.5,1.5) {};
    \node (h2) at (7.5,1.5) {};
    \node (i2) at (8.5,1.5) {};
    \draw[line width=1pt, myRed] (g1) -- (g2);
    \draw[line width=5pt, myRed] (h1) -- (h2) node[draw=none, fill=none, midway, inner sep=0pt, left] {$M_2^{''*}$};
    \draw[line width=1pt, myRed] (i1) -- (i2);

    \draw[bm,line width=5pt] (f1) -- (g2)  node[draw=none, fill=none, pos=0.6, inner sep=0pt, above left] {$M_2^{'*}$};

    \draw (h1) -- (i2);

    \draw[bmtwo,line width=3pt] (i1) -- (do5) node[draw=none, fill=none, midway, right] {$M_1^{'*}$};
    \draw[bmtwo] (h1) -- (do4);
    \draw[bmtwo] (h2) -- (up5);
    \draw[bmtwo,line width=3pt] (i2) -- (up5)  node[draw=none, fill=none, midway, right] {$M_1^{'*}$};

    \begin{scope}[on background layer]
    \fill[SecondaryLightBlue!55] plot [smooth cycle] coordinates {(9.75,2.6) (7,3.5) (1.5,3.5) (1.5,2.5) (5.5,2) (5.5,-0.5) (1.5,-1) (1.5,-2) (7,-2) (9.55,-1)};
    \end{scope}

    \draw[dashed, black, thick] plot [smooth cycle] coordinates {(5.0, -1.8)  (7.0, 3.1) (0.1, 3.1) (0.5, -1.7)};
\end{tikzpicture}
\vspace{-6em}
\caption{Illustration for the proof of \cref{lem:apx051bipartite}.
The thick edges belong to a fixed maximum matching $M^*$.
Each of them is labeled with its partition ($M_2^*$, $M_2^{'*}$, $M_2^{''*}$, $M_1^*$, or $M_1^{'*}$).
The two subgraphs for which we invoke \Cref{lem:585} are marked
(case 1 -- highlighted in light blue, case 2 -- dashed line).
}
    \label{fig2}
\end{figure}
    
    Since $M \cup M'$ is maximal, $M^*$ cannot have edges with no endpoints in $V(M) \cup V(M')$.
    We thus have 
    \begin{equation}
        \mu(G) = |M^*| = \mo + \mop + \mt + \mtp + \mtpp . \label{eq1}
    \end{equation}
    Also, by maximality of $M^*$ and simple counting,
    \begin{align}
        |M| &= \mt + \frac12 \mo + \frac12 \mtp . \label{eq2}
    \end{align}
    We will use \cref{lem:585} to analyze both cases in \Cref{alg:sublinear-first}.
    For case 1,
    we can use $G' = G[V \setminus V(M)]$;
    in this graph, $M'$ is a maximal matching,
    and $B_1$ is a $b$-matching as in the statement of \cref{lem:585}
    (see \cref{fig2}), thus we have
    \begin{equation*}
       (1 - \frac{1}{b}) |M'| + \frac{1}{kb} |B_1| \ge (2-\sqrt{2}-\epsilon) \mu(G[V \setminus V(M)]) \ge (2-\sqrt{2}-\epsilon) (\mop + \mtpp)
    \end{equation*}
    since $M_1^{'*} \cup M_2^{''*}$ is a matching in $G' = G[V \setminus V(M)]$.
    With \cref{eq2} this gives
    \begin{equation}
        |M| + (1 - \frac{1}{b}) |M'| + \frac{1}{kb} |B_1| \ge \mt + \frac12 \mo + \frac12 \mtp + (2-\sqrt{2}-\epsilon) (\mop + \mtpp). \label{eqcase1}
    \end{equation}
    For case 2,
    we can instead use $G' = G[V(M)] \cup G[V(M), V \setminus V(M)]$;
    in this graph, $M$ is a maximal matching,
    and $B_2$ is a $b$-matching as in the statement of \cref{lem:585}
    (see \cref{fig2}),
    thus
    \begin{equation}
        (1 - \frac{1}{b}) |M| + \frac{1}{kb} |B_2| \ge (2-\sqrt{2}-\epsilon) \mu(G') \ge  (2-\sqrt{2}-\epsilon) (\mt + \mtp + \mo) \label{eqcase2}
    \end{equation}
    since $M_2^{*} \cup M_2^{'*} \cup M_1^*$ is a matching in $G'$.
    Using \cref{eqcase1,eqcase2}, we can bound the left-hand side of this lemma's statement as
    \begin{align*}
        \max \left[...\right] &\ge \max [
        \mt + \frac12 \mo + \frac12 \mtp + (2-\sqrt{2}-\epsilon) (\mop + \mtpp), \\
        &\qquad \quad \ \; (2-\sqrt{2}-\epsilon) (\mt + \mtp + \mo)
        ] \\
        &\ge ...
    \end{align*}
    We bound the maximum by a weighted average with weights $\beta$ and $1-\beta$ for some $\beta \in [0,1]$ to be determined ($\max(x,y) \ge \beta x + (1-\beta) y$):
    \begin{align*}
    ... &\ge \mt (\beta + (2-\sqrt{2}-\epsilon)(1-\beta)) \\&\qquad + (\mo + \mtp) \left( \frac{\beta}{2} + (2-\sqrt{2}-\epsilon)(1-\beta) \right) \\&\qquad + (\mop + \mtpp) (2-\sqrt{2}-\epsilon) \beta
    \\
    &\ge 
    (\mt + \mo + \mtp) \left( \frac{\beta}{2} + (2-\sqrt{2}-\epsilon)(1-\beta) \right) + (\mop + \mtpp) (2-\sqrt{2}-\epsilon) \beta
    \\
    &\stackrel{(*)}{\ge}
    (\mt + \mo + \mtp) \gamma + (\mop + \mtpp) \gamma
    \\
    &\stackrel{\eqref{eq1}}{=} \gamma \cdot \mu(G).
    \end{align*}
    We want $(*)$ to hold for some $\gamma$ (the approximation ratio) as large as possible.
    For $(*)$ to hold, we need to satisfy:
    \begin{align*}
        \frac{\beta}{2} + (2-\sqrt{2}-\epsilon) (1-\beta) &\ge \gamma,\\
        (2-\sqrt{2}-\epsilon) \beta &\ge \gamma.
    \end{align*}
    By solving for $\frac{\beta}{2} + (2-\sqrt{2}) (1-\beta) = (2-\sqrt{2}) \beta$ we get $\beta = \frac{12 + 2 \sqrt{2}}{17}$ and $\gamma = \frac{4(5-2\sqrt{2})}{17} - O(\epsilon) > 0.5109$.
\end{proof}

\begin{restatable}{lemma}{approxguarantee}
\label{lem:approx-guarantee}
    Let $\max(\mu_1, \mu_2)$ be the output of \Cref{alg:sublinear-first}. With high probability, it holds that
    \begin{align*}
        0.5109\cdot \mu(G) - o(n) \leq \max(\mu_1, \mu_2) \leq \mu(G).
    \end{align*}
\end{restatable}
The proof of \cref{lem:approx-guarantee} is routine.
\begin{proof}
    By \Cref{lem:apx051bipartite}, we have
    \begin{align}\label{eq:range-estimate}
        0.5109  \cdot\mu(G) \leq \max \left[|M| + (1 - \frac{1}{b}) \E|M'| + \frac{1}{kb} \E|B_1|, (1 - \frac{1}{b}) |M| + \frac{1}{kb} \E|B_2| \right] \leq \mu(G).
    \end{align}
    Let $X_i$, $Y_i$, and $Z_i$ be as defined in \Cref{alg:first-case} and \Cref{alg:second-case}. By the definition of $X_i$, $Y_i$, and $Z_i$ we have
    \begin{align*}
        \E[X_i] = \Pr[X_i = 1] = \frac{2\E|M'|}{n},\\
        \E[Y_i] = \Pr[Y_i = 1] = \frac{2\E|B_1|}{n},\\
        \E[Z_i] = \Pr[Z_i = 1] = \frac{2\E|B_2|}{n}.
    \end{align*}
    Thus,
    \begin{align*}
        \E[X] = \frac{2r \cdot \E|M'|}{n},\\
        \E[Y] = \frac{2r \cdot \E|B_1|}{n},\\
        \E[Z] = \frac{2r \cdot \E|B_2|}{n}.
    \end{align*}
    Then, using Chernoff bound, we obtain
    \begin{align*}
        \Pr[|X - \E[X]| \geq \sqrt{6\E[X]\log n}] \leq 2\exp\left( \frac{6 \E[X] \log n}{3\E[X]}\right) = \frac{2}{n^2},\\
        \Pr[|Y - \E[Y]| \geq \sqrt{6\E[Y]\log n}] \leq 2\exp\left( \frac{6 \E[Y] \log n}{3\E[Y]}\right) = \frac{2}{n^2},\\
        \Pr[|Z - \E[Z]| \geq \sqrt{6\E[Z]\log n}] \leq 2\exp\left( \frac{6 \E[Z] \log n}{3\E[Z]}\right) = \frac{2}{n^2}.
    \end{align*}
    Therefore, with a probability of $1-2/n^2$,
    \begin{align*}
        \mu_{M'} = \frac{nX}{2r} - \frac{n}{2\log n} &\in \frac{n(\E[X] \pm \sqrt{6\E[X]\log n})}{2r} - \frac{n}{2\log n}\\
        &= \E|M'| - \frac{n}{2\log n} \pm \sqrt{\frac{3n \E|M'| \log n}{2r}}\\
        & = \E|M'| - \frac{n}{2\log n} \pm \sqrt{\frac{n \E|M'|}{4\log^2 n}} & (\text{Since  $r = 6 \log^3 n$})\\
        & = \E|M'| - \frac{n}{2\log n} \pm \frac{n}{2\log n}.
    \end{align*}
    By repeating the same argument for $Y$ and $Z$, we get
    \begin{align*}
        \mu_{B_1} \in  \E|B_1| - \frac{n}{2\log n} \pm \frac{n}{2\log n}, \\
        \mu_{B_2} \in  \E|B_2| - \frac{n}{2\log n} \pm \frac{n}{2\log n}.
    \end{align*}
    Plugging \Cref{eq:range-estimate} in the bounds obtained above implies
    \begin{align*}
        \max(\mu_1, \mu_2) &= \max \left[|M| + (1 - \frac{1}{b}) \mu_{M'} + \frac{1}{kb} \mu_{B_1}, (1 - \frac{1}{b}) |M| + \frac{1}{kb} \mu_{B_2} \right]\\
        & \leq \max \left[|M| + (1 - \frac{1}{b}) \E|M'| + \frac{1}{kb} \E|B_1|, (1 - \frac{1}{b}) |M| + \frac{1}{kb} \E|B_2| \right]\\
        & \leq \mu(G).
    \end{align*}
    On the other hand, we have
    \begin{align*}
        \max(\mu_1, \mu_2) &= \max \left[|M| + (1 - \frac{1}{b}) \mu_{M'} + \frac{1}{kb} \mu_{B_1}, (1 - \frac{1}{b}) |M| + \frac{1}{kb} \mu_{B_2} \right]\\
        & \geq \max \left[|M| + (1 - \frac{1}{b}) (\E|M'| - \frac{n}{\log n}) + \frac{1}{kb} (\E|B_1|- \frac{n}{\log n}), (1 - \frac{1}{b}) |M| + \frac{1}{kb} (\E|B_2|- \frac{n}{\log n}) \right]\\
        & \geq \max \left[|M| + (1 - \frac{1}{b}) \E|M'|  + \frac{1}{kb} \E|B_1|, (1 - \frac{1}{b}) |M| + \frac{1}{kb} \E|B_2| \right] - \frac{2n}{\log n}\\
        & \geq 0.5109\cdot \mu(G) - \frac{2n}{\log n}
    \end{align*}
    which completes the proof.
\end{proof}

\subsection{Running Time Analysis}

For the analysis of the running time, we use a crucial property of random greedy maximal matching algorithm that was proved recently in \cite{steiner-tree-itcs}.

\begin{proposition}[Lemma 5.14 of \cite{steiner-tree-itcs}]\label{prop:outdegree-bound-1}
    Let $Q(v)$ be the expected number of times that the oracle queries an adjacent edge of $v$ if we start the oracle calls from a random vertex, for a random permutation over the edges of the graph $G$ when running random greedy maximal matching. It holds that $Q(v) = \wt{O}(\deg_G(v) / |V(G)|)$.
\end{proposition}

\begin{proposition}[Corollary 5.15 of \cite{steiner-tree-itcs}]\label{prop:outdegree-bound}
    Let $T(v)$ be the expected time needed to return a random neighbor of vertex $v$. Then, the expected time to run the random greedy maximal matching oracle for a random vertex and a random permutation in graph $G$ is $\sum_{v \in V(G)} \wt{O}(T(v) \cdot \deg_G(v) / |V(G)|)$.
\end{proposition}

\begin{lemma}\label{lem:first-alg-time}
    \Cref{alg:first-case} runs in $\widetilde{O}( \bar{d}(G) \cdot \sqrt{n})$ time in expectation.
\end{lemma}
\begin{proof}
    Without loss of generality, we assume that $|V \setminus V(M)| = \Omega(n)$; otherwise, \Cref{alg:sublinear-first} does not need to execute \Cref{alg:first-case}, as the estimate from \Cref{alg:second-case} suffices.

    The proof consists of two parts: first we show that line 6 of the algorithm can be implemented in $\wt{O}(\bar{d}(G))$ expected time, and then we prove that line 8 of the algorithm can be implemented in $\wt{O}(\bar{d}(G)\cdot \sqrt{n})$ expected time.

    For the first part, let $G'' = G[V \setminus V(M)]$ and let $v$ be a vertex in $G''$. In the adjacency list of $v$ in the original graph, which contains $\deg_G(v)$ elements, only $\deg_{G''}(v)$ of them are neighbors in $G''$. Thus, to find a neighbor in $G''$, we need to randomly sample $T(v) = \deg_G(v) / \deg_{G''}(v)$ elements from $v$'s adjacency list in expectation. Consequently, using \Cref{prop:outdegree-bound}, the expected time required to execute the random greedy maximal matching oracle for a randomly chosen vertex and permutation in $G''$ is
    \begin{align*}
        &\sum_{v \in V(G'')} 
    \wt{O}\left(\frac{T(v) \cdot \deg_{G''}(v)}{|V(G'')|}\right) = \sum_{v \in V(G'')} 
    \wt{O}\left(\frac{(\deg_G(v)/\deg_{G''}(v)) \cdot \deg_{G''}(v)}{|V(G'')|}\right)\\
    &\qquad \qquad = \sum_{v \in V(G'')} 
    \wt{O}\left(\frac{\deg_G(v)}{| V(G'')|}\right)
    = \wt{O}\left(\frac{|E(G)|}{|V(G'')|}\right)
    = \wt{O}\left(\frac{|E(G)|}{|V(G)|}\right)
     = \widetilde{O}(\bar{d}(G)) .
    \end{align*}

    In the algorithm, we choose a permutation $\pi$ and run the random greedy maximal matching for $r = \wt{O}(1)$ randomly chosen vertices. By linearity of expectation, the expected running time of the first part is $\widetilde{O}(\bar{d}(G))$.

    We prove the second part in a few steps. As a first step, for simplicity, assume that we have access to the adjacency list of $G'_1$ and there is no need to run nested oracles. Then, using \Cref{prop:rgmm}, the expected running time of line 8 is bounded by $\wt{O}(\bar{d}(G'_1))$ for $r = \wt{O}(1)$ random permutations and random vertices. 
    
    Now, suppose that we have access to the adjacency list of $G'' = G[V \setminus V(M)]$ instead of $G'_1$. By \cref{prop:outdegree-bound},
    the total expected runtime of the outer ($b$-matching) oracle is
    $$\wt{O}\left( \sum_{v \in G'_1} \Touter(v) \cdot \deg_{G'_1}(v) / |V(G'_1)|\right),$$
    where $\Touter(v)$ is the expected time needed to find a random neighbor of $v$ in $G_1'$.
    To find such a neighbor, we sample neighbors $w$ of $v$ in $G''$
    and check whether $w$ is matched in $M'$ using the inner oracle.
    The expected number of these checks until a vertex $w$ with $(v,w) \in E(G_1')$ is found is $\wt{O}(\deg_{G''}(v)/\deg_{G'_1}(v))$.
    The expected cost of each check is $\E_{w : (v,w) \in E(G'')}[\costinner(w)]$,
    where $\costinner(w)$ is the runtime of invoking the inner oracle to check if $w$ is matched in $M'$.
    Putting this together, the expected total runtime is
    \begin{align*}
        \wt{O}\left(\frac{\sum_{v \in G'_1} \Touter(v) \cdot \deg_{G'_1}(v)}{|V(G'_1)|} \right)
        &=
        \wt{O}\left(\frac{1}{|V(G'_1)|} \sum_{v \in G'_1} \frac{\deg_{G''}(v)}{\deg_{G'_1}(v)} \cdot \frac{\sum_{w : (v,w) \in E(G'')}\costinner(w)}{\deg_{G''}(v)} \cdot \deg_{G'_1}(v)\right)
        \\
        &=
        \wt{O}\left(\frac{1}{|V(G'_1)|} \sum_{w \in G''} \costinner(w) \cdot \deg_{G''}(w)\right)
        \\
        &=
        \wt{O}\left(\sum_{w \in G''} \frac{\costinner(w)}{|V(G'')|} \cdot \deg_{G''}(w)\right)
        \\
        &\stackrel{\Cref{lem:sparsification-maximal}}{=}
        \wt{O}\left(\sqrt{n} \cdot \E_{w \in G''} [\costinner(w)]\right)
        \\
        &\stackrel{\Cref{prop:outdegree-bound}}{=}
        \wt{O}\left(\sqrt{n} \cdot \sum_{v \in G''} \deg_{G''}(v) \cdot \Tinner(v) \cdot \frac{1}{|V(G'')|}\right)
    \end{align*}
    where $\Tinner(v)$ is the expected time needed to find a random neighbor of $v$ in $G''$.
    Under our assumption of having access to the adjacency list of $G''$, $\Tinner(v) = O(1)$
    and thus we can finish bounding with $\wt{O}(\sqrt{n} \cdot \bar{d}(G'')) = \wt{O}(\sqrt{n} \cdot \bar{d}(G))$.

    Now we lift the assumption of having access to the adjacency list of $G''$.
    This means that, similarly as in the first part, we need to sample $\Tinner(v) = \deg_G(v) / \deg_{G''}(v)$ vertices from the adjacency list of $v$ in expectation (checking if a vertex is matched in $M$ is done in $O(1)$ time).
    Then the total bound becomes
    \begin{align*}
    \wt{O}\left(\sqrt{n} \cdot \sum_{v \in G''} \deg_{G''}(v) \cdot \frac{\deg_G(v)}{\deg_{G''}(v)} \cdot \frac{1}{|V(G'')|}\right)
    = \wt{O}\left(\sqrt{n} \cdot \frac{|E(G)|}{|V(G'')|}\right)
    = \widetilde{O}(\sqrt{n} \cdot \bar{d}(G)) .
    \end{align*}
    Also $\Touter(v)$ increases by $\deg_G(v) / \deg_{G_1'}(v)$,
    which increases the total runtime only by
    \begin{align*}
        \wt{O}\left(\frac{1}{|V(G'_1)|} \sum_{v \in G'_1} \frac{\deg_G(v)}{\deg_{G'_1}(v)} \cdot \deg_{G'_1}(v) \right) = \wt{O}(\bar{d}(G)).
    \end{align*}
\end{proof}

\begin{lemma}\label{lem:second-alg-time}
    \Cref{alg:second-case} runs in $\widetilde{O}(\bar{d}(G))$ time in expectation.
\end{lemma}
\begin{proof}
    It is enough to show that it takes $\wt{O}(\bar{d}(G))$ time to run the algorithm of \Cref{prop:rgmm} (the $b$-matching oracle) for each sampled vertex and permutation, as we have $r = \wt{O}(1)$. 
    We repeat the same arguments as in the proof of~\Cref{lem:first-alg-time}.
    
    Let $v$ be a vertex in the graph $G'_2 = G[V(M), V\setminus V(M)]$. Note that in the adjacency list of $v$ in the original graph, which contains $\deg_G(v)$ elements, only $\deg_{G'_2}(v)$ of the elements are neighbors in $G'_2$. Thus, we need to randomly sample $T(v) = \deg_G(v) / \deg_{G'_2}(v)$ elements of the adjacency list of $v$ to find a neighbor in $G'_2$.
    Recall that we can check whether a vertex is matched in $M$ in $O(1)$ time.
    Therefore, using \Cref{prop:outdegree-bound}, the expected time needed to run the random greedy maximal matching for a random vertex and permutation in $G'_2$ is
    \begin{align*}
        &\sum_{v \in V(G'_2)} 
    \wt{O}\left(\frac{T(v) \cdot \deg_{G'_2}(v)}{|V(G'_2)|}\right) = \sum_{v \in V} 
    \wt{O}\left(\frac{(\deg_G(v)/\deg_{G'_2}(v)) \cdot \deg_{G'_2}(v)}{n}\right)\\
    &\qquad \qquad = \sum_{v \in V} 
    \wt{O}\left(\frac{\deg_G(v)}{n}\right)
    = \wt{O}\left(\frac{|E(G)|}{n}\right)
    = \widetilde{O}(\bar{d}(G)).
    \end{align*}%
\end{proof}

\begin{lemma}\label{lem:final-time-1}
    \Cref{alg:sublinear-first} runs in $\widetilde{O}(n\sqrt{n})$ time with high probability.
\end{lemma}
\begin{proof}
    By \Cref{clm:sparsification-time}, the sparsification step takes $\wt{O}(n\sqrt{n})$ time. Also, by \Cref{lem:first-alg-time,lem:second-alg-time}, \Cref{alg:first-case,alg:second-case} run in $\widetilde{O}( \bar{d}(G) \cdot \sqrt{n})$ and $\widetilde{O}(\bar{d}(G))$ expected time, respectively. To establish a high probability bound on the time complexity, we execute $O(\log n)$ instances of the algorithm in parallel and halt as soon as the first one completes. By applying Markov’s inequality, we deduce that each individual instance terminates within $\widetilde{O}(n\sqrt{n})$ time with constant probability. As a result, with high probability, at least one of these instances finishes within $\widetilde{O}(n\sqrt{n})$ time. This concludes the proof.
\end{proof}

Now we are ready to prove the final theorem of this section.

\begin{theorem}
    There exists an algorithm that, given access to the adjacency list of a bipartite graph, estimates the size of the maximum matching with a multiplicative-additive approximation factor of $(0.5109, o(n))$ and runs in $\widetilde{O}(n\sqrt{n})$ time with high probability.
\end{theorem}
\begin{proof}
    By \Cref{lem:approx-guarantee}, \Cref{alg:sublinear-first} achieves multiplicative-additive approximation of $(0.5109, o(n))$. Moreover, the running time of \Cref{alg:sublinear-first} is $\widetilde{O}(n\sqrt{n})$ with high probability by \Cref{lem:final-time-1}.
\end{proof}




\section{Algorithm for General Graphs}\label{sec:general}

The following is an analogue of \cref{lem:apx051bipartite} for general graphs.
\begin{lemma} \label{lem:apx051general}
    In a general graph $G$,
    let $M$, $M'$, $B_1$ and $B_2$ be as in the description of \Cref{alg:sublinear-first}. Then
    \[(1-\epsilon)\max \left[|M| + \E[\mu(M' \cup B_1)], \E[\mu(M \cup B_2)] \right] \geq 0.5109  \cdot\mu(G) .\]
\end{lemma}
\begin{proof}
    The proof proceeds as that of \cref{lem:apx051bipartite}, with the following modifications:
    \begin{itemize}
        \item we invoke the second rather than the first part of \cref{lem:585} to obtain analogues of inequalities \eqref{eqcase1} and \eqref{eqcase2},
        i.e., instead of lower-bounding $|M| + (1 - \frac{1}{b}) |M'| + \frac{1}{kb} |B_1|$, we lower-bound $|M| + \E[\mu(M' \cup B_1)]$,
        and instead of lower-bounding
        $(1 - \frac{1}{b}) |M| + \frac{1}{kb} |B_2|$,
        we lower-bound $\E[\mu(M \cup B_2)]$,
        \item the $(1-\epsilon)$ factor on the left-hand side of the statement can be folded into the $O(\epsilon)$ term in $\gamma$ at the end of the proof, for small enough $\epsilon$.
    \end{itemize}
\end{proof}

In this section, we show how to extend our algorithm to work for general graphs. The main difference between bipartite and general graphs is that the estimate based on the sizes of $|M|$, $|M'|$, $|B_1|$, and $|B_2|$ is insufficient to achieve a 0.5109 approximation guarantee. In \Cref{lem:apx051general}, we show that we can achieve this approximation ratio by estimating $\mu(M \cup B_2)$ and $\mu(M' \cup B_1)$. More formally, to produce $\mu_1$ and $\mu_2$ in \Cref{alg:sublinear-first}, we use $|M| + \mu(M' \cup B_1)$ and $\mu(M \cup B_2)$, respectively. Also, \Cref{alg:first-case} and \Cref{alg:second-case} output $\mu(M' \cup B_1)$ and $\mu(M \cup B_2)$, respectively.

In both \Cref{alg:first-case} and \Cref{alg:second-case} we have access to oracles that can return whether a vertex is matched in $M'$, $B_1$, or $B_2$ for a fixed permutation $\pi$. These oracles can also be used to return the edge of the matching if the vertex is matched, which is a corollary of \Cref{prop:rgmm} since the algorithm of \Cref{prop:rgmm} can be also used to return the edge of the matching.

\begin{lemma}\label{lem:edge-oracle-1}
    For a vertex $v$, there exists an algorithm that returns the edges of $v$ in $M'$ and $B_1$ in $\wt{O}(\bar{d}(G) \cdot \sqrt{n})$ expected time.
\end{lemma}
\begin{proof}
    The proof follows from \Cref{lem:first-alg-time} and the fact that the algorithm in \Cref{prop:rgmm} can return the edge of the matched vertex in the same running time.
\end{proof}

\begin{lemma}\label{lem:edge-oracle-2}
    For a vertex $v$, there exists an algorithm that returns the edges of $v$ in $M$ and $B_2$ in $\wt{O}(\bar{d}(G))$ expected time.
\end{lemma}
\begin{proof}
    If $v$ is matched in $M$, we can return the edge in $O(1)$ time since we have access to this maximal matching. The rest of the proof follows from \Cref{lem:second-alg-time} and the fact that the algorithm in \Cref{prop:rgmm} can return the edge of the matched vertex in the same running time.
\end{proof}

Local computation algorithms (LCA) is a model of computation, also motivated by large data sets, in which the algorithm is not expected to produce the entire output at once. Instead, the algorithm is queried for parts of the output, and must produce a consistent and approximately optimal output. We use the following local computation algorithm (LCA) by \cite{LeviRY17} to design our algorithm. 

\begin{proposition}[\cite{LeviRY17}]\label{prop:lca}
    There exists a $(1-\epsilon)$-approximate local computation algorithm for maximum matching of graph $G$ in $\wt{O}(\Delta(G)^{1/\epsilon^2})$ time with access to the adjacency list of $G$.
\end{proposition}

Now we prove the main technical part of this section that can be used to estimate $\mu(M' \cup B_1)$ and $\mu(M \cup B_2)$.

\begin{lemma}\label{lem:matching-sparse-levi}
    Let $H$ be a subgraph of graph $G$. Suppose that $H$ is the union of a constant number of random greedy maximal matching on different subsets of vertices. Also, we have oracle access to edges of random greedy maximal matching. We can query a vertex to obtain the matching edge of vertex $v$ in $\wt{O}(T)$ expected time.  Moreover, the maximum degree of $H$ is constant. Then, there exists a $(1-\epsilon)$-approximate algorithm that estimates the size of the maximum matching of $H$ in $\wt{O}(T)$ expected time.
\end{lemma}
\begin{proof}
    We run the algorithm of \Cref{prop:lca}. Each time the algorithm visits a new vertex $u$, we first query all the constant number of oracles for random greedy maximal matching to get the adjacency list of $u$ in $H$. If the algorithm in \Cref{prop:lca} made uniform queries to the oracle, then we could conclude the proof since $\Delta(H)$ and $\epsilon$ are constant. However, note that queries to the oracles are not independent and we have an expected time guarantee on the running time of the oracles. So it is possible that the way that the algorithm of \Cref{prop:lca} works might result in making biased queries to oracles.

    Let $\mathcal{L}(u, \pi)$ be the set of vertices the algorithm of \Cref{prop:lca} visits when we run the algorithm for vertex $u$ and we use permutation $\pi$ for random greedy maximal matchings of the subgraph $H$. Let $F(u,\pi)$ be the running time of the oracle for vertex $u$ and permutation $\pi$. Also, let $L(u, \pi)$ be the running time of the algorithm of \Cref{prop:lca} on vertex $u$ using permutation $\pi$ for random greedy maximal matchings. Hence, we have
    \begin{align*}
        \E_{u,\pi}[F'(u,\pi)] = \sum_\pi \sum_u \sum_{v \in \mathcal{L}(u, \pi)} \frac{\E[F(v,\pi)]}{n\cdot|E(G)|!}
    \end{align*}
    Further, the algorithm in \Cref{prop:lca} explores the local neighborhood of a vertex to answer each query. Therefore, if two vertices $w$ and $z$ are at a distance greater than $\Delta(H)^{1/\epsilon^3}$, the algorithm does not visit $z$ when answering a query for $w$. Thus,
    \begin{align*}
        \E_{u,\pi}[F'(u,\pi)] &\leq \sum_\pi \sum_v \Delta(H)^{\left(\Delta(H)^{1/\epsilon^3}\right)}\cdot \frac{\E[F(v,\pi)]}{n\cdot|E(G)|!}\\
        & = \Delta(H)^{\left(\Delta(H)^{1/\epsilon^3}\right)} \cdot \sum_\pi \sum_v \frac{\E[F(v,\pi)]}{n\cdot|E(G)|!}\\
        & \leq \wt{O}(T),
    \end{align*}
    where the first inequality follows by the fact that each vertex has at most $\Delta(H)^{\left(\Delta(H)^{1/\epsilon^3}\right)}$ neighbors in $H$ within a distance of $\Delta(H)^{1/\epsilon^3}$, and the last inequality because $\Delta(H)$ and $\epsilon$ are constants. Therefore, even with the biased queries, we can implement \Cref{prop:lca} on subgraph $H$ in $\wt{O}(T)$ time, which completes the proof.
\end{proof}

\begin{lemma}\label{lem:muMB}
    There exists an algorithm that outputs a $(1-\epsilon)$-approximate estimation of the value of $\mu(M \cup B_2)$ in $\wt{O}(\bar{d}(G))$ expected time.
\end{lemma}
\begin{proof}
    The proof follows by plugging \Cref{lem:edge-oracle-2} into \Cref{lem:matching-sparse-levi} and running the algorithm for $r$ samples.
\end{proof}

\begin{lemma}\label{lem:muMMBB}
    There exists an algorithm that outputs a $(1-\epsilon)$-approximate estimation of the value of $\mu(M' \cup B_1)$ in $\wt{O}(\bar{d}(G) \cdot \sqrt{n})$ expected time.
\end{lemma}
\begin{proof}
    The proof follows by plugging \Cref{lem:edge-oracle-1} into \Cref{lem:matching-sparse-levi} and running the algorithm for $r$ samples.
\end{proof}

\begin{theorem}
    There exists an algorithm that, given access to the adjacency list of a graph, estimates the size of the maximum matching with a multiplicative-additive approximation factor of $(0.5109, o(n))$ and runs in $\widetilde{O}(n\sqrt{n})$ time with high probability.
\end{theorem}
\begin{proof}
    By \Cref{lem:apx051general} and a Chernoff bound similar to \Cref{lem:approx-guarantee}, we achieve a multiplicative-additive approximation of $(0.5109, o(n))$. Moreover, the expected running time is $\wt{O}(n\sqrt{n})$ by \Cref{clm:sparsification-time}, \Cref{lem:muMB}, and \Cref{lem:muMMBB}. To establish a high probability bound on the time complexity, we execute $O(\log n)$ instances of the algorithm in parallel and halt as soon as the first one completes. By applying Markov’s inequality, we deduce that each individual instance terminates within $O(n\sqrt{n})$ time with constant probability. As a result, with high probability, at least one of these instances finishes within $O(n\sqrt{n})$ time. This concludes the proof.
\end{proof}

\section{Multiplicative Approximation}\label{sec:multiplicative}
In this section, we show that we can achieve a multiplicative approximation guarantee by slightly increasing the number of samples in \Cref{alg:first-case} and \Cref{alg:second-case}. First, we prove a simple lower bound for the size of the maximum matching of a graph based on its maximum and average degree.

\begin{claim}\label{clm:mu-delta-d}
    For any graph $G$, it holds that $\mu(G) \geq n\bar{d}(G) / (4\Delta(G))$.
\end{claim}
\begin{proof}
    Any graph with a maximum degree of $\Delta(G)$ can be colored greedily using $2\Delta(G)$ colors. Furthermore, the edges of each color create a matching. Thus, we have $\mu(G) \geq |E(G)| / (2\Delta(G))$. Combining with $|E(G)| = n\bar{d}(G)/2$, we can conclude the proof.
\end{proof}

The goal is to obtain a multiplicative approximation guarantee of $(0.5109 - \epsilon) \mu(G)$. It is important to note that if any of \( |M| \), \( |M'| \), \( |B_1| \), or \( |B_2| \) is not a constant fraction of the others, it can be omitted from the equation in the statement of \Cref{lem:apx051bipartite} without affecting the approximation by more than a function of \( \epsilon \). Thus, without loss of generality, we can assume that $|M| = \Omega(\mu(G))$, $|M'| = \Omega(\mu(G))$, $|B_1| = \Omega(\mu(G))$, and $|B_2| = \Omega(\mu(G))$. Consequently, using \Cref{clm:mu-delta-d} and as an application of Chernoff bound, we can use $\wt{O}_\epsilon(\Delta(G)/\bar{d}(G))$ samples in \Cref{alg:first-case} and \Cref{alg:second-case} to obtain multiplicative estimation of $\mu_{M'}$, $\mu_{B_1}$, and $\mu_{B_2}$.

By \Cref{lem:first-alg-time}, \Cref{alg:first-case} runs in $\wt{O}(\bar{d}(G) \cdot \sqrt{n})$ time when we have $r = \wt{O}(1)$ samples. By increasing the number of samples to $\wt{O}_\epsilon(\Delta(G)/\bar{d}(G))$, the running time of \Cref{alg:first-case} increases to $\wt{O}(\Delta(G) \cdot \sqrt{n})$. Moreover, by \Cref{lem:second-alg-time}, \Cref{alg:second-case} runs in $\wt{O}(\bar{d}(G))$ time when we have $r = \wt{O}(1)$ samples. By increasing the number of samples to $\wt{O}_\epsilon(\Delta(G)/\bar{d}(G))$, the running time of \Cref{alg:second-case} increases to $\wt{O}(\Delta(G))$. Therefore, the total running time of the algorithm is within $\wt{O}(n \sqrt{n})$.

Finally, we can obtain the degree of each vertex in the graph using binary search. Therefore, we can assume that we have access to $\Delta(G)$ and $\bar{d}(G)$ by spending $\wt{O}(n)$ time.
Thus we get:

\maintheoremAdjlist*

\section{Algorithm with Access to the Adjacency Matrix}\label{sec:matrix}

In this section, we use a simple reduction to show that with a small modification, our algorithm  can be adapted to the setting where we have access to the graph's adjacency matrix. A slightly different version of this kind of reduction appeared in previous works on sublinear time algorithms for maximum matching \cite{Behnezhad21, BehnezhadRRS23}.

It is important to note that obtaining a constant-factor multiplicative approximation is impossible when the algorithm only has access to the adjacency matrix of the graph. This is because if the graph is guaranteed to contain either a single edge or be completely empty, any algorithm would require $\Omega(n^2)$ adjacency matrix queries to distinguish between these two cases. Consequently, we allow the algorithm to have an additive error of $o(n)$ in addition to the multiplicative approximation ratio.

We build an auxiliary graph $H$ with the following vertex set and edge set:
\begin{itemize}
    \item \textbf{Vertex set:} $V(H)$ contains $n+2$ disjoint sets of $n$ vertices $V_1, V_2$, and $U_1, \ldots, U_n$. Each $V_i$ is a copy of the vertices of the original graph. Each $U_i$ contains $n \log^2 n$ vertices.
    \item \textbf{Edge set:} For each vertex $v \in V_1$, the $i$-th neighbor of $v$ is the $i$-th vertex in $V_1$ if $(v, i) \in E(G)$, and otherwise it is the $i$-th vertex in $V_2$. Similarly, for each vertex $v \in V_2$, the $i$-th neighbor of $v$ is $i$-th vertex in $V_2$ if $(v, i) \in E(G)$, and otherwise it is the $i$-th vertex in $V_1$. Also, each $v \in V_2$ is connected to all vertices of $U_v$. As a result, the degree of each vertex in $U_1 \cup U_2 \cup \ldots \cup U_n$ is 1, the degree of each vertex in $V_2$ is $n$, and the degree of each vertex in $V_2$ is $n + n\log^2 n$. Therefore, we have $\Delta(H) = \wt{O}(n)$.
\end{itemize}

Because of the way we constructed the graph, it is not hard to see that we can find the $i$-th neighbor of the adjacency list of each vertex in $H$ using only a single query to the adjacency matrix of $G$.

\begin{observation}
    For each vertex $v$ in graph $H$, the $i$-th neighbor of $H$ can be found using at most a single adjacency matrix query in $G$.
\end{observation}
\begin{proof}
    If $i > \deg_H(v)$, we can simply certify this since the degree of each vertex in $H$ is determined by the vertex set to which $v$ belongs. If $v \in U_j$, then the only neighbor of $v$ is $j$ in $V_2$. If $v \in V_1$, we make an adjacency matrix query for $(v, i)$, and based on the result, the $i$-th neighbor is either vertex $i$ in $V_1$ or in $V_2$. A similar procedure applies for $v \in V_2$ when $i \leq n$. If $i > n$ for $v \in V_2$, we return the $(i - n)$-th vertex in $U_v$.  
\end{proof}

\paragraph{Modification to the algorithm:} Since the graph contains $\wt{\Theta}(n^2)$ vertices, we cannot afford to apply the sparsification step to all vertices. However, vertices in $U_1 \cup \ldots \cup U_n$ have degree 1. Therefore, we apply the sparsification step only to vertices in $V_1$ and $V_2$. Since we have $|V_1| + |V_2| = 2n$, we can apply the sparsification for these sets in $\wt{O}(n\sqrt{n})$ time. We first iterate over the vertices in $V_2$ and apply the sparsification step, and then we apply it to the vertices in $V_1$. This ordering ensures that most vertices in $V_2$ get matched to vertices in $U_1 \cup \ldots \cup U_n$ in this step, which is desirable for our application.

\begin{claim}
    After the sparsification step, each vertex in $V_2$ is matched by $M$ with high probability. Moreover, at most $n/\log n$ vertices in $V_2$ are matched to vertices in $V_1 \cup V_2$ with high probability.
\end{claim}
\begin{proof}
    Note that we first process the vertices in $V_2$. At the time we process a vertex $v \in V_2$, it has at least $n\log^2 n$ neighbors due to its neighbors in $U_v$. Since, after sparsification, each vertex must have a degree of $\wt{O}(\sqrt{n})$ with high probability by \Cref{lem:sparsification-maximal}, all vertices in $V_2$ will be matched with high probability.

    Additionally, at the time we process a vertex $v \in V_2$, it has $n\log^2 n$ unmatched neighbors in $U_v$. On the other hand, it has at most $n$ neighbors in the rest of the graph. Thus, with a probability of at most $1/\log^2 n$, it gets matched to a vertex outside $U_v$. Therefore, in expectation, at most $n/\log^2 n$ vertices in $V_2$ are matched to vertices in $V_1 \cup V_2$. Using the Chernoff bound, we can conclude that at most $n/\log n$ vertices in $V_2$ are matched to vertices in $V_1 \cup V_2$ with high probability.
\end{proof}

Equipped with this reduction, we can now simply run the rest of the algorithm for vertices in $V_1$. The only difference is that we exclude the edges of $M$ that lie between $V_2$ and $U_1 \cup \ldots \cup U_n$ in the estimation. Additionally, in the final estimation, the algorithm returns the previous estimate minus $n/\log n$, accounting for the vertices in $V_2$ that are not matched within $V_1 \cup V_2$, which introduces an additional $o(n)$ additive error.
Thus we obtain:

\maintheoremAdjmat*

\bibliographystyle{alpha}
\bibliography{reference}

\end{document}